\documentclass[11 pt]{article}

\pdfoutput=1

\usepackage{graphicx}
\usepackage{amsmath}
\usepackage{amssymb}
\usepackage{enumerate}
\usepackage{amsthm}
\usepackage{dsfont}
\usepackage[T1]{fontenc}
\usepackage[latin1]{inputenc}
\usepackage{latexsym}
\usepackage{amssymb}
\usepackage{mathrsfs}
\usepackage{pdfsync}
\usepackage[latin1]{inputenc}
\usepackage[english]{babel}
\usepackage{color}
\usepackage{subfig}

\DeclareMathOperator\sign{sign}

\renewcommand{\mathcal}{\mathscr}

\newcommand{\fX}{\mathfrak{X}}

\newcommand{\EE}{\mathbb{E}}

\newcommand{\NN}{\mathbb{N}}

\newcommand{\RR}{\mathbb{R}}

\newcommand{\XX}{\mathbb{X}}
\newcommand{\YY}{\mathbb{Y}}

\renewcommand{\epsilon}{\varepsilon}
\renewcommand{\rho}{\varrho}

\newcommand{\set}{\triangleq}

\newcommand{\norm}[1]{ \left\| #1 \right\|}

\newtheorem{thm}{Theorem}[section]
\newtheorem{prop}[thm]{Proposition}
\newtheorem{cor}[thm]{Corollary}
\newtheorem{example}[thm]{Example}
\newtheorem{assumption}[thm]{Assumption}
\newtheorem{lem}[thm]{Lemma}
\newtheorem{remark}[thm]{Remark}
\theoremstyle{definition}

\newenvironment{rem}{\begin{remark}} { $\hfill \diamond$ \end{remark}}

\numberwithin{equation}{section}

\begin{document}

\renewcommand{\baselinestretch}{1.2} \normalsize

\title{Illiquidity and Derivative Valuation }
\author{Ulrich Horst, Felix Naujokat \\ \small Institut f\"ur Mathematik \\ \small Humboldt-Universit\"at Berlin \\ \small horst@mathematik.hu-berlin.de  \\ \small naujokat@mathematik.hu-berlin.de }
\date{\today}
\maketitle

\begin{abstract}
    In illiquid markets, option traders may have an incentive to increase their portfolio value by using their impact on the dynamics of the underlying. We provide a mathematical framework within which to value derivatives under market impact in a multi-player framework by introducing strategic interactions into the Almgren \& Chriss (2001) model. Specifically, we consider a financial market model with several strategically interacting players that hold European contingent claims and whose trading decisions have an impact on the price evolution of the underlying. We establish existence and uniqueness of equilibrium results and show that the equilibrium dynamics can be characterized in terms of a coupled system of possibly non-linear PDEs. For the linear cost function used in Almgren \& Chriss (2001), we obtain (semi) closed form solutions for risk neutral or CARA investors. Finally, we indicate how spread crossing costs discourage market manipulation.
\end{abstract}

\vspace{2mm}

\begin{center}
{\sc Preliminary Version - Comments Welcome}
\end{center}

\mbox{ }



\vspace{2mm}


\newpage

    \section{Introduction}

Standard financial market models assume that asset prices follow an
exogenous stochastic process and that all transactions can be
settled at the prevailing price without any impact on market
dynamics. The assumption that all trades can be carried out at
exogenously given prices is appropriate for small investors that
trade only a negligible proportion of the overall daily trading
volume; it is not appropriate for institutional investors trading large blocks of shares over a short time span. The trading decisions of institutional investors are likely to move stock prices in an unfavorable direction and often induce significant trading costs. 

It is now widely recognized that (the lack of) liquidity is a major source of financial risk and there has been an increasing interest in mathematical models of illiquid financial markets.
Much of the literature on illiquidity focusses on either
optimal hedging and portfolio liquidation strategies for a single
large investor under market impact (Cetin, Jarrow \& Protter (2004), Alfonsi, Fruth \& Schied (2007), Rogers \& Singh (2007)),
predatory trading (Carlin, Lobo \& Viswanathan (2007), Schoeneborn \& Schied (2008)) and the role
of derivative securities including
the problem of market manipulation using options (Jarrow (1992), Kumar \& Seppi (1992)).
In an illiquid market, derivative traders have an
incentive to utilize their impact on the price dynamics of the
underlying in order move the option value in a favorable
direction\footnote{Gallmeyer \& Seppi (2000) provide some evidence
that in illiquid markets option traders are in fact able to increase
a derivative's value by moving the price of the underlying.}. It has been shown by Jarrow
(1994), for instance, that by introducing derivatives into an otherwise complete and
arbitrage-free market, manipulation strategies with a risk free gain
may appear, such as market corners and front runs. 
Schoenbucher \& Wilmott (2000) discuss an illiquid market
model where a large trader can influence the stock price with
vanishing costs and risk. They argue that the risk of manipulation
on the part of the large trader makes the small traders unwilling to
trade derivatives any more. In particular, they predict that the
option market breaks down. Our analysis indicates that markets do not
necessarily break down when stock price manipulation is costly as it
is in our model. 

While the aforementioned papers differ significantly in their degree of complexity, they all focus on a {\sl single player} framework. When multiple players are considered the analysis is typically confined to some form of ''stealth trading'' as in Carlin, Lobo \& Viswanathan (2007) and Schoeneborn \& Schied (2008), where liquidity providers try to benefit from the liquidity demand that comes from some ''large'' investor but no strategic interaction between liquidity suppliers and consumers is considered. In fact, so far only little work has been devoted to models with strategically interacting market participants. Vanden (2005) considers a pricing game in continuous time where the option issuer controls the volatility of the underlying but does not
incur liquidity or spread crossing costs. He derives a Nash equilibrium in the two player,
risk neutral case and shows that ''seemingly harmless derivatives,
such as ordinary bull spreads, offer incentives for manipulation
that are identical to those offered by digital options'' (p. 1892,
l. 36). Gallmeyer \& Seppi (2000) consider a
binomial model with three periods and finitely many risk neutral agents
holding call options on an illiquid underlying. Assuming a linear
permanent price impact and linear transaction costs, and assuming
that all agents are initially endowed with the same derivative they
prove the existence of Nash equilibrium trading strategy and
indicate how market manipulation can be reduced. 


We provide a general mathematical framework within which to value derivative
securities in illiquid markets under strategic interactions thereby  
extending the work of Gallmeyer \& Seppi (2000) in several directions. 
Specifically, we consider a pricing game between a finite number of
large investors (``players'') holding European claims written on an
illiquid stock. Their goal is to maximize expected utility at
maturity from trading the stock where their portfolio value
at maturity depends on the trading strategies of all the other
players, due to their impact on the dynamics of the underlying.
Following Almgren \& Chriss (2001) we assume that the players have a
permanent impact on stock prices and that all trades are settled at
the prevailing market price plus a liquidity premium. The liquidity
premium can be viewed as an instantaneous price impact that affects
transaction prices but not the value of the players' inventory. This
form of market impact modeling is analytically more tractable than
that of Obizhaeva \& Wang (2006) which also allows for temporary
price impacts and resilience effects. It has also been
adopted by, e.g., Carlin, Lobo \& Viswanathan (2007), Schoeneborn \& Schied (2008)
and many practitioners from the financial industry.

Our framework is flexible enough to allow for rather general
liquidity costs including the linear cost function of Almgren \&
Chriss (2001) and some form of spread crossing costs. We show that
when the market participants are risk neutral or have CARA utility
functions the pricing game has a unique Nash equilibrium in the
class of absolutely continuous trading strategies; existence results
for more general utility functions are given for the one player
case. We solve the problem of equilibrium pricing using techniques
from the theory of stochastic optimal control and stochastic
differential games. We show that the family of the players' value
functions can be characterized as the solution to a coupled system of
non-linear PDEs. Here we use a-priori estimates for Nash equilibria;
we prove that the system of PDEs has a unique classical solution
with bounded derivatives. It turns out that the equilibrium problem
can be solved in closed form for a specific market environment,
namely the linear cost structure used in Almgren \& Chriss (2001)
and risk neutral agents. We use this explicit solutions to state
some conditions which make manipulation unattractive or avoid it altogether.
For instance, we show that when the agents are risk neutral no market manipulation occurs in zero sum games, i.e., in a game between an option writer and an option issuer. Furthermore, we find that
the bid ask spread is important determinant of market manipulation. It turns out that the higher the spread, the less beneficial market manipulation: high spread
crossing costs make trading more costly and hence discourage
frequent re-balancing of portfolio positions.

This paper is organized as follows: We present the market model in
section 2. In chapter 3, we formulate the optimization problem,
derive a priori estimates for Nash equilibria and prove the
existence of a solution for one player with general utility
function. We solve the multi-player case in section 4 for risk
neutral and CARA agents. We use these solutions in chapter 5 to show
how market manipulation can be reduced. Section 6 concludes.

    \section{ The Model }

We adopt the market impact model of Schoeneborn \& Schied (2007)
with a finite set $J$ of {\sl agents}, or {\sl players}, trading a single stock whose price
process depends on the agent's trading strategies. Following Almgren
\& Chriss (2001) we shall assume that the players have a permanent
impact on asset prices and that all trades are settled at prevailing
market prices plus a liquidity premium which depends on the {\sl
change} in the players' portfolios. In order to be able to capture
changes in portfolio positions in an analytically tractable way, we
follow Almgren \& Chriss (2001) and Schoeneborn \& Schied (2007)
and restrict ourselves to absolutely continuous trading strategies.
Hence we consider only trading strategies from the class
    $$ \fX \set \left\{ X : [0,T] \mapsto \RR | X \text{ absolutely continuous, adapted and } X_0 = 0   \right\} $$
where we denote by $X^j_t$ the number of stock
shares held by player $j \in J \set \{ 1, ..., N \}$ at time $t \in
[0, T] $. We write $d X^j_t = \dot X^j_t dt$ and call $\dot X^j$ the
{\sl trading speed} of the player $j$.

%

\subsection{Price dynamics and the liquidity premium}

Our focus is on valuation schemes for derivatives with short
maturities under strategic market interactions. For short trading
periods it is appropriate to model the {\sl fundamental stock
price}, i.e., the value of the stock in the absence of any market
impact, as a Brownian Motion with volatility $(\sigma B_t)$. Market
impact is accounted for by assuming that the investors' accumulated
stock holdings $\sum_{i=1}^N X^i $ have a linear impact on the stock
process $(P_t)$ so that
    \begin{equation} \label{preisdynamik1}
        P_t = P_0 + \sigma B_t + \lambda \sum_{i=1}^N X^i_t
    \end{equation}
with a permanent impact parameter $\lambda > 0$. The linear
permanent impact is consistent with the work of Huberman \& Stanzl
(2004) who argued that linearity of the permanent price impact is
important to exclude quasi-arbitrage\footnote{There is some
empirical evidence that very large trades have a concave price
impact but this observation needs further validation.}.

A trade at time $t \in [0,T]$ is settled at a {\sl transaction
price} $\tilde P_t$ that includes an additional instantaneous price
impact, or {\sl liquidity premium}. Specifically,
\begin{equation} \label{preisdynamik2}
        \tilde P_t = P_t + g \left( \sum_{i=1}^N \dot X^i_t \right)
\end{equation}
with a {\sl cost function} $g$ that depends on the
instantaneous change $\sum_{i=1}^N \dot X^i $ in the agents'
position in a possibly non-linear manner. The liquidity premium
accounts for limited available liquidity, transaction costs, fees,
spread crossing costs, etc. Spread-crossing costs are of particular
importance and have not been considered in the previous literature
on market impact.

\begin{rem}
In our model the liquidity costs are the same for all traders and
depend only on the aggregate demand throughout the entire set of
agents. This captures situations where the agents trade through a
market maker or clearing house that reduces the trading costs by
collecting all orders and matching incoming demand and supply prior
to settling the outstanding balance $\sum_{i=1}^N \dot X^i_t$ at
market prices.
\end{rem}

We assume with no loss of generality that $g$ is normalized, i.e.,
$g(0) = 0$ and smooth. The following additional mild assumptions on $g$
will guarantee that the equilibrium pricing problem has a solution
for risk neutral and CARA investors.

\begin{assumption} \label{assumption_g}
\begin{itemize}
    \item The derivative $g'$ is bounded away from zero, that is $g' > \epsilon >
    0$.
    \item The mapping $z \mapsto g(z) +  z g'(z) $ is strictly
    increasing.
\end{itemize}
\end{assumption}
The first condition is natural for a cost function. Since  the
liquidity costs associated with a net change in the overall position
$z$ is given by $z g(z)$, the second assumption states that the
agents face increasing marginal costs of trading. Our assumptions on
$g$ are satisfied for the following important examples:

\begin{example}
Cost functions which satisfy Assumption \ref{assumption_g} are the
linear cost function $g(z) = \kappa z$ with $\kappa > 0$, used in
Almgren \& Chriss (2001) and cost functions of the form
\[
    g(z) = \kappa z + s \frac{2}{\pi} \arctan(C z) \quad \mbox{with} \quad s, C >
    0.
\]
The former is the cost function associated with a block-shaped limit
order book. The latter can be viewed as a smooth approximation of the
map $z \mapsto \kappa z + s \cdot \sign(z)$ which is the cost
function associated with a block-shaped limit order book and
spread $s > 0$.
\end{example}

\subsection{Preferences and endowments}

Each agent is initially endowed with a contingent claim $H^j = H^j
(P_T)$, whose payoff depends on the stock price $P_T$ at maturity.
Although it is not always necessary we assume that the functions
$H^j$ are smooth and bounded with bounded derivatives $H^j_p$.

\begin{rem}
We only consider options with cash delivery. The assumption of cash
delivery is key. While cash settlement is susceptible to market
manipulation, we show in Section 5 below that when deals are settled
physically, i.e., when the option issuer delivers the underlying,
market manipulation is not beneficial: the cost of acquiring at an
increased price outweighs the benefits from a possible higher option
payoff, due to an increase in the underlying.
\end{rem}

We model the risk preference of the agent $j \in J$ with a von
Neumann - Morgenstern utility function $u^j$ and assume that her aim
is to maximize her expected portfolio value at maturity from trading
in the financial market so the agent's optimization problem is given
by:
    \begin{equation} \label{orig_problem}
        \sup_{ X^j \in \fX } \EE \left[ u^j \left( - \int_0^T \dot X^j_t \tilde P_t dt + H^j(P_T) + value^j(X^j_T) \right) \right]
        .
    \end{equation}
The portfolio value consists of the trading costs $- \int_0^T \dot
X^j_t \tilde P_t dt$, the option payoff $H^j(P_T)$ and the
liquidation value $value^j(X^j_T)$ of $X^j_T$ stock shares at
maturity. Rigorously defining a form of liquidation value in a
multi-player framework is a question of its own mathematical
interest, and is not the focus of this paper. The problem of defining a liquidation value in a single-player framework has been solved in a recent paper by Schoeneborn \& Schied (2008); in a
game-theoretic setting the problem is much more involved.
The agents optimize against their beliefs about the other
players' individual assessments of their respective portfolio values. In order to simplify the analysis we work under Assumption \ref{assumption} below. 

\begin{assumption} \label{assumption}
All agents optimize their utility assuming that
\begin{equation} \label{value}
    value^j(X^j_T) = \int_0^T \dot X^j_t P_t dt.
\end{equation}
\end{assumption}

The preceding assumption is motivated by the single player, risk-neutral framework where (\ref{value}) holds in expected values. In such a setting it states that the expected costs (utility) of building up a portfolio over the time span $[0,T]$ under market impact equals the expected liquidation costs (utility) under infinitely slow liquidation. In a multi-player model Assumption \ref{assumption} is only a first benchmark that simplifies the subsequent analysis even if all player are risk-neutral. The assumption is nonetheless useful as it allows us to carry out our analysis of equilibrium and to derive some insight into the structure of market dynamics under strategic interactions. 

Under Assumption \ref{assumption} the individual optimization problem
(\ref{orig_problem}) reduces to
\begin{equation} \label{problem}
    \sup_{ X^j \in \fX } \EE \left[ u^j \left( - \int_0^T \dot X^j_t g
    \left( \sum_{i = 1}^N \dot X^i_t \right) dt + H^j(P_T) \right)
    \right] .
\end{equation}
We say that a vector of strategies $\left( \dot X^1, ..., \dot X^N
\right) $ is a \textit{Nash equilibrium} if for each agent $j \in J$ her
trading strategy $X^j$ is a best response against the behavior of
all the other players, i.e., if $X^j$ solves (\ref{problem}), given
the vector $X^{-j} \set (X^i)_{i \neq j}$. In the following section we
derive a-priori estimates for equilibrium trading strategies and use
standard results from the theory of stochastic optimization to show
that Nash equilibria can be characterized in terms of a coupled
system of partial differential equations (PDEs). For the special
case of risk neutral and CARA investors we show that the system of
PDEs has a solution so that a unique (in a certain class)
equilibrium exists.

\section{Equilibrium Dynamics and A-Priori Estimates}
In this section we formulate the optimization problem
(\ref{problem}) as a stochastic control problem, derive the
associated Hamilton-Jacobi-Bellman-equations, HJB for short, and
transform it into a system of coupled PDEs. To this end, we choose
the stock price $P$ and the trading costs $R^j$ of the agent $j \in
J$ as state variables. They evolve according to:
\begin{eqnarray}
    d P_t &=& \sigma d B_t + \lambda \sum_{i=1}^N \dot X^i_t dt, \quad P_0 = p_0 \label{state1} \\
    d R^j_t &=& \dot X^j_t g \left( \sum_{i=1}^N \dot X^i_t \right) dt, \quad R^j_0 = 0. \label{state2}
\end{eqnarray}
For a given time $t < T$, spot price $p$ and trading costs $r$ the
{\sl value function} of the player $j$, defined by
    \begin{equation} \label{valuefunction}
    V^j (t, p, r) \set \sup_{ X^j \in \fX } \EE_t \left[ u^j \left( - r^j - \int_t^T \dot X^j_s g \left( \sum_{i = 1}^N \dot X^i_s \right) ds + H^j(P_T) \right)  \mid P_t = p
    \right],
    \end{equation}
denotes the maximal expected portfolio value at maturity that the player
can achieve by trading the underlying. The associated HJB-equation
is
\begin{equation} \label{HJB}
    \left\{
    \begin{aligned}
        0 & = v^j_t + \frac{1}{2} \sigma^2 v^j_{pp} +
        	\sup_{c^j} \left[ \lambda \left( c^j + \dot X^{-j} \right) v^j_p +
        	c^j g \left( c^j + \dot X^{-j} \right) v^j_{r^j}
      	 \right] \\
        v^j(T, p, r) & = u^j \left( - r^j + H^j(p) \right)
    \end{aligned}
    \right.
\end{equation}
cf. Fleming \& Soner (1993). Here we have used $v^j_{r^i} \equiv 0$
for $i \neq j$. Given the trading strategies $X^{-j}$ of all the
other agents, a candidate for the maximizer $c^j = \dot X^j $ should
satisfy
\begin{equation} \label{opt_cond}
    0 = \lambda \frac{v_p^j}{v^j_{r^j}} + g \left( c^j + \dot X^{-j} \right) + c^j g' \left( c^j + \dot X^{-j} \right)
\end{equation}
provided $v^j_{r^j} \neq 0$. We sum up these equations over the set
of players in order to get the following characterization for the
cumulated equilibrium trading speed $ \sum_{i = 1}^N \dot X^i_t $:
\begin{equation} \label{opt_cond_sum}
    0 = \lambda \sum_{i=1}^N \frac{v_p^i}{v^i_{r^i}} + N g \left(
    \sum_{i = 1}^N \dot X^i_t \right) + \left( \sum_{i = 1}^N \dot X^i_t
    \right) g' \left( \sum_{i = 1}^N \dot X^i_t \right).
\end{equation}
Due to Assumption \ref{assumption_g}, $z \mapsto N g(z) + z g'(z)$ is strictly increasing. Hence, equation
(\ref{opt_cond_sum}) admits a unique solution $ \dot X^* = \dot X^*
\left( \sum_{i=1}^N \frac{v_p^i}{v^i_{r^i}} \right)$. Plugging this
solution back into (\ref{opt_cond}) allows to compute the optimal
strategy $c^j = \dot X^j$ in terms of $v^i_p$ and $v^i_{r^i}$ as
\begin{equation} \label{strategy_j}
    c^j = \dot X^j = - \frac{1}{g' \left( \dot X^* \right) } \left[ \lambda \frac{v^j_p}{v^j_{r^j}} + g ( \dot X^* ) \right] .
\end{equation}
This expression is well defined if $v^j_{r^j} \neq 0$ because $g'
> 0$. To conclude, we have turned the family of individual HJB-equations (\ref{HJB}) into
the following system of coupled PDEs for $j = 1, ..., N$:
\begin{equation} \label{PDE_system}
    \left\{
    \begin{aligned}
    0 & = v^j_t + \frac{1}{2} \sigma^2 v^j_{pp} + \lambda \dot X^* v^j_p +
    \dot X^j g \left( \dot X^* \right) v^j_{r^j} \\
    v^j(T, p, r) & = u^j \left( - r^j + H^j(p) \right)
    \end{aligned}
    \right.
\end{equation}
where the coupling stems from the aggregate trading speed $ \dot X^*
= \dot X^* \left( \sum_{i=1}^N \frac{v_p^i}{v^i_{r^i}} \right)$.

Solving the system (\ref{PDE_system}) is delicate, to say the least.
The problem is the non-linearity coming from the expression $\dot X^j g
\left( \dot X^* \right)$ along with the implicit dependence of $\dot
X^j$ on the derivatives $v^i_p$ and $v^i_{r^i}$ and the fact
that $v^i_{r^i}$ appears in the denominator. The latter problem can
be coped with by assuming that (after a possible monotone
transformation) the agents' risk preferences satisfy a translation
property so that $u_r = 1$. A large class of such utility functions
can be linked to backward stochastic differential equations but we
choose not to embed our work into that line of research. We consider
instead the case of risk neutral and CARA investors where the
existence of a unique classical solution to the system
(\ref{PDE_system}) is guaranteed without any reference to backward
equations. The proof uses the following a-priori
estimates for the optimal trading strategies. It states that in equilibrium the trading speed is bounded. As a result, the agents' utilities from trading and the value function associated with their respective HJB equations along with (as we shall see) their derivatives are bounded.


\begin{lem} \label{boundedstrategy}
    Let $ \left( \dot X^1, ..., \dot X^N \right) $ be a Nash equilibrium
    for problem (\ref{problem}). Then each strategy $\dot X^j$ satisfies
\[
    \left| \dot X^j \right| \leq  N \frac{\lambda}{\epsilon} \max_i \norm
    {H^i_p}_{\infty}.
\]
\end{lem}

    \begin{proof}
    Let $j \in J$, $h \set \max_i \norm {H^i_p}_{\infty} $ and
    $ A \set \left\{ \sum_{i=1}^N \dot X^i \geq 0 \right\} $ be the set where the aggregate trading speed is nonnegative.
    Let us fix the sum of the competitors' strategies $\dot X^{-j}$. 
    On the set $A$ the best response $\dot X^j$ is bounded above
    by $ K \set \frac{\lambda}{\varepsilon} h $. Otherwise the truncated strategy
    $ \dot Y^j \set \dot X^j \wedge \frac{\lambda}{\varepsilon} h $ would outperform $\dot X^j$. To see this, let us compare the
    payoffs associated with $\dot X^j$ and $\dot Y^j$. The payoff associated with $X$ minus the payoff associated with $Y$ can estimated from below as
    \begin{eqnarray*}
        &    &          - \int_0^T \dot Y^j g \left( \dot Y^j + \dot X^{-j} \right) dt + H^j (P_T (Y^j)) \\
        &    &          + \int_0^T \dot X^j g \left( \dot X^j + \dot X^{-j} \right) dt - H^j (P_T (X^j)) \\
        &\geq&  \int_0^T \dot Y^j \left( g \left( \dot X^j + \dot X^{-j} \right) - g \left( \dot Y^j + \dot X^{-j} \right) \right) \\
        &    &  + \left( \dot X^j - \dot Y^j \right) g \left( \dot X^j + \dot X^{-j} \right) dt -
        \lambda (X_T^j - Y_T^j) \norm {H_p}_{\infty}.
    \end{eqnarray*}
    Note that $ \dot X^j + \dot X^{-j} \geq 0$ on $A$ and thus
    $ g \left( \dot X^j + \dot X^{-j} \right) \geq 0 $ due to Assumption \ref{assumption_g}.
    Furthermore, $ g \left( \dot X^j + \dot X^{-j} \right) - g \left( \dot Y^j
    + \dot X^{-j} \right) \geq \epsilon \left( \dot X^j - \dot Y^j \right) $,
    again by Assumption \ref{assumption_g}. The difference in the
    payoffs is therefore strictly bigger than
    \begin{eqnarray*}
        &  &  \int_0^T \dot Y^j \epsilon \left( \dot X^j - \dot Y^j \right) dt - \lambda h \int_0^T \dot X^j - \dot Y^j dt \\
        &\geq& \int_{ \dot X^j > \dot Y^j} \left( \epsilon \dot Y^j - \lambda h \right)  \left( \dot X^j - \dot Y^j \right) dt \\
        & = & 0.
    \end{eqnarray*}
    This shows that $\dot X^j$ is bounded above on the set $A$. A symmetric argument shows that on the complement $ A^c $ the optimal response is bounded below by $- K$. Furthermore, this implies that on $A^c$ the optimal response is bounded above by
    $$ \dot X^j = \sum_{i=1}^N \dot X^i + \sum_{i \neq j} - \dot X^i \leq 0 + (N-1) K $$
A similar argument yields that $\dot X^j$ is bounded below on $A$. This completes the proof.
\end{proof}

    \begin{rem}
        These a priori estimates, together with the boundedness assumption on the payoffs $H^j$, imply that the optimization problem (\ref{problem}) is bounded. Moreover, if a smooth solution to the PDE system (\ref{PDE_system}) exists, it is bounded.
    \end{rem}

In the one player framework we can use a standard result from the theory of stochastic control to show that (\ref{PDE_system}) admits a unique solution.
\begin{prop} \label{gen_ut_1p}
    Let $N=1$ and for the terminal condition $\psi(p, r) \set u(-r + H(p))$ let  $\psi \in \mathcal C^3 $ and let $\psi, \psi_p, \psi_r$ satisfy a polynomial growth condition. Then the Cauchy problem (\ref{PDE_system}) admits a unique classical solution in $ \mathcal C^{1,2} $, which coincides with the value function $V$.
\end{prop}

\begin{proof}
    Due to Lemma (\ref{boundedstrategy}), the optimal control $\dot X$ can be chosen from a compact set. Thus, we can apply Theorem $IV.4.3$ in Fleming \& Soner (1993), which yields that (\ref{PDE_system}) admits a unique solution in $\mathcal C^{1,2}$, which is of polynomial growth. It remains to apply the  Verification Theorem $IV.3.1$ from Fleming \& Soner (1993) to see that this solution coincides with the agent's value function $V$.
\end{proof}

%
%

\section{Examples}

In this section we establish existence and uniqueness of equilibrium
results for risk neutral and CARA investors. For risk neutral
investors and linear cost functions the equilibrium strategies can
be given in closed form; if spread crossing costs are involved a
closed form solution is not available and we report numerical
results instead.

\subsection{Risk Neutral Agents} 
Let us assume that all players are risk neutral, i.e. $u^j (z) = z$. In a first step we prove existence of a unique solution to the system (\ref{PDE_system}) for general cost functions $g$. Subsequently we construct an explicit solution to (\ref{PDE_system}) for the linear cost structure used in Algmen \& Chriss (2001).

    \subsubsection{General Cost Structure}
Let $g$ be a general cost function which satisfies Assumption \ref{assumption_g}. In the risk neutral case the value function of player $j$ turns into
    $$ V^j(t, p, r) = - r^j + \sup_{X^j \in \fX } \EE_t \left[ - \int_t^T  \dot X^j_s g \left( \sum_{i=1}^N \dot X^i_s \right) ds + H^j( P_T ) | P_t = p \right] . $$
In particular, $V^j_{r^j} \equiv -1$ so the optimal strategies do not depend on the trading costs. In other words, the state variable $r$ is redundant and we omit it in this section. We write
    $$ V^j(t, p) \set V^j(t, p, 0) . $$
The HJB-equation (\ref{HJB}) turns into
    \begin{equation} \label{HJB_rn}
        0 = v^j_t + \frac{1}{2} \sigma^2 v^j_{pp} + \sup_{c^j} \left[ \lambda \left( c^j + \dot X^{-j} \right) v^j_p - c^j g \left( c^j + \dot X^{-j} \right)  \right]
    \end{equation}
where we have used $ v^j_{r^j} = -1 $. The optimal trading speed form (\ref{strategy_j}) is given by
    \begin{equation} \label{strategy_j_rn}
        c^j = \dot X^j = - \frac{1}{g' \left( \dot X^* \right) } \left[ - \lambda v^j_p + g ( \dot X^* ) \right]
    \end{equation}
where the aggregate trading speed $\dot X^* = \sum_{i=1}^N \dot X^j$ is the unique solution to
    \begin{equation} \label{cumspeed_rn}
        0 = \lambda \sum_{i=1}^N v_p^i - N g \left( \sum_{i = 1}^N \dot X^i_t \right) - \left( \sum_{i = 1}^N \dot X^i_t \right) g' \left( \sum_{i =1}^N \dot X^i_t \right) .
    \end{equation}
The system of PDEs (\ref{PDE_system}) therefore takes the form
        \begin{equation} \label{system_rn}
            0 = v^j_t + \frac{1}{2} \sigma^2 v^j_{pp} + \lambda \dot X^* v^j_p - \dot X^j g \left( \dot X^* \right)
        \end{equation}
with terminal condition $v^j(T, p) = H^j(p)$. The following proposition shows that a unique solution exists if $H \in \mathcal C^2_b$, i.e. $H$ and its derivatives up to order $2$ are bounded. The proof follows from a general existence result stated in Appendix A.

\begin{prop} \label{prop_existence_rn}
    Let $H \in \mathcal C^2_b$. Then the Cauchy problem (\ref{system_rn}) admits a unique classical solution in $ \mathcal C^{1,2} $, which coincides with the vector of value functions.
\end{prop}

An alternative way of solving the system (\ref{system_rn}) is the following: If we sum up the $N$ equations, we get a Cauchy problem for the aggregate value function $v \set \sum_{i=1}^N v^i$, namely
\[
    0 = v_t + \frac{1}{2} \sigma^2 v_{pp} + \dot X^* \left[ \lambda v_p - g \left( \dot X^* \right) \right]
\]
with terminal condition $v (T, p) = \sum_{i=1}^N H^i (p)$. Existence and uniqueness of a solution to this one-dimensional problem can be shown using Theorem IV.8.1 in Ladyzenskaja (1968). Once the solution is known, we can plug it back into (\ref{system_rn}) and get $N$ decoupled equations. This technique is applied in the following section where we construct an explicit solution for linear cost functions.

    \subsubsection{Linear Cost Structure} \label{riskneutral}
    For the particular choice $g(z) = \kappa z$ $(\kappa > 0)$ used in Algmren \& Chriss (2001) and Schoeneborn \& Schied (2007), the solution to (\ref{system_rn}) can be given explicitly.

    \begin{cor}
       Let $g(z) = \kappa z$. Then the solution of (\ref{system_rn}) can be given in closed form as the solution to a nonhomogeneous heat equation.
    \end{cor}

    \begin{proof}
    The optimal trading speed from (\ref{strategy_j_rn}) and the aggregate trading speed from (\ref{cumspeed_rn}) are
    \begin{eqnarray}
        \dot X^j &=& \frac{\lambda}{\kappa} \left( v_p^j - \frac{1}{N+1} \sum_{i=1}^N v_p^i \right) \label{speed_j} \\
        \dot X^* &=& \sum_{i=1}^N \dot X^i = \frac{\lambda}{\kappa (N+1)} \sum_{i=1}^N v^i_p = \frac{\lambda}{\kappa (N+1)} v_p . \label{speed}
    \end{eqnarray}
Let us sum up the $N$ equations from (\ref{system_rn}). This yields the following PDE for the aggregate value function $v = \sum_{i=1}^N v^i$:
    \begin{equation} \label{cum_solution_rn}
        0 = v_t + \frac{1}{2} \sigma^2 v_{pp} + \frac{ \lambda^2 N }{ \kappa (N+1)^2 } v_p^2
    \end{equation}
with terminal condition $ v(T,p) = \sum_{i=1}^N H^i(p) $. This PDE is a variant of \textit{Burgers' equation}, cf. Rosencrans (1972). It allows for an explicit solution, which we cite in Lemma \ref{Burgers}. With this solution at hand, we can solve for each single investor's value function. We plug the solution $v$ back into the equations (\ref{speed_j}) and (\ref{speed}) for the trading speeds, and those into the PDE (\ref{system_rn}). This yields
        $$ 0 = v^j_t + \frac{1}{2} \sigma^2 v^j_{pp} + \frac{\lambda^2}{\kappa (N+1)^2} v_p^2$$
with terminal condition $v^j(T, p) = H^j(p)$. This nonhomogeneous heat equation is solved by
    $$ v^j (T-t, p) = \int_{\RR} H^j d \mathcal N(p,\sigma t ) + \frac{ \lambda^2 }{ \kappa (N+1)^2 } \int_0^t \int_{\RR} v_p^2 (s, \cdot) d \mathcal N (p, \sigma (t-s) ) $$
where $v$ is given in Lemma \ref{Burgers} and $\mathcal N$ denotes the heat kernel. 
    \end{proof}

In the preceding proof and in Corollary \ref{cor_cara} we need the solution to a variant of \textit{Burgers' equation}. We cite it in the following Lemma.
    \begin{lem} \label{Burgers}
        Let $A > 0, B \neq 0$. The PDE
            $$ 0 = 2 v_t + A v_{pp} + B v_p^2 $$
        with terminal value
            $$ v(T, p) = G(p) $$
        is solved by
            $$ v(t,p) = \frac{ A }{B} \log \left[ \int_{\RR} \exp \left( \frac{B}{ A } G \left( \sqrt{A} z \right) \right) d \mathcal N \left( \frac{p}{\sqrt{A}}, T-t \right) \right] . $$
    \end{lem}

    \begin{proof}
        By means of a linear transformation we can reduce the problem to $A=B=1$. This particular case is solved in Rosencrans (1972).
    \end{proof}

    \subsubsection{Numerical Illustrations}
In the risk neutral setting, we were able to reduce the system of
PDEs from the multi-player setting to the one-dimensional PDE
(\ref{cum_solution_rn}) for the aggregate value function. This can
be interpreted as the value function of the representative agent.
Such reduction to a representative agent is not always possible for
more general utility functions. In the sequel we illustrate the
optimal trading speed $\dot X(t, p)$ and surplus\footnote{By surplus
we difference between the representative agent's optimal expected
utility $v(t, p)$ and the conditional expected payoff $\EE_t[H( P_T
) | P_t = p]$ in the absence of any market impact.} of a
representative agent as  functions of time and spot prices for a
European call options $H(P_T) = (P_T - K)^+$ and digital option
$H(P_T) = \mathds 1_{ \{ P_T \geq K \} }$, respectively. We choose a
linear cost function, strike $K= 100$, maturity $T = 1$, volatility
$\sigma = 1$ and liquidity parameters $\lambda = \kappa = 0.01$. We
see from Figure \ref{bild_1p_rn_call} that for the case of a call
option both the optimal trading speed and the surplus increases with
the spot; the latter also increases with the time to maturity.
Furthermore, the increase in the trading speed if maximal, when the
option is at the money. For digital options the trading speed is
highest for at the money options close to maturity as the trader
tries to push the sport above the strike. If the spot is far away
from the strike, the trading speed is very small as it is unlikely
that the trader can push the sport above the strike 
before expiry.

    \begin{figure}
        \includegraphics[width = 0.45 \textwidth]{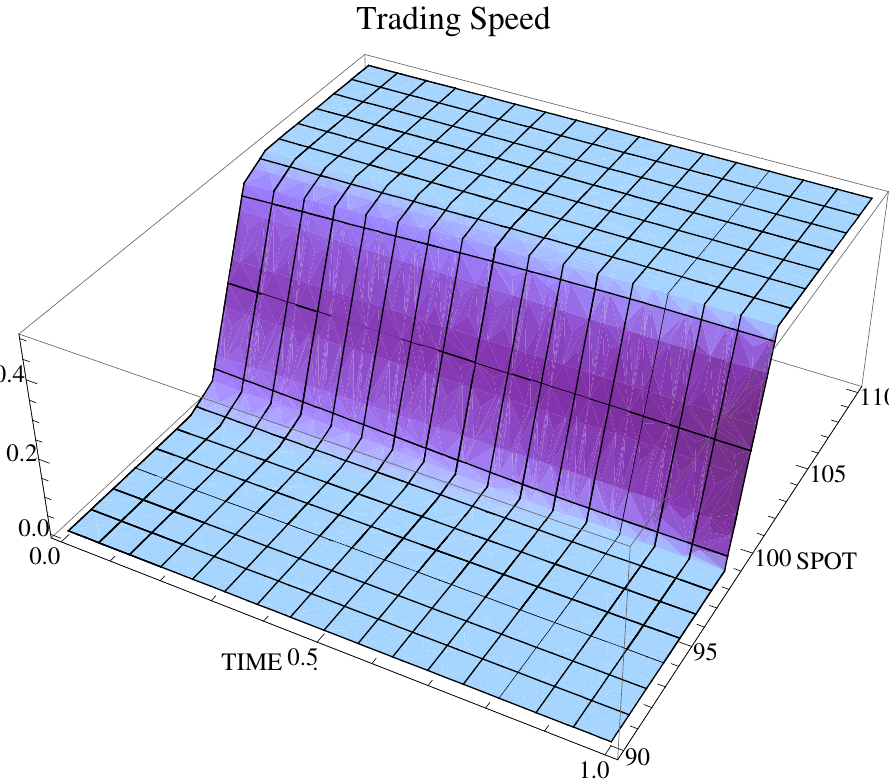} \hfill
        \includegraphics[width = 0.45 \textwidth]{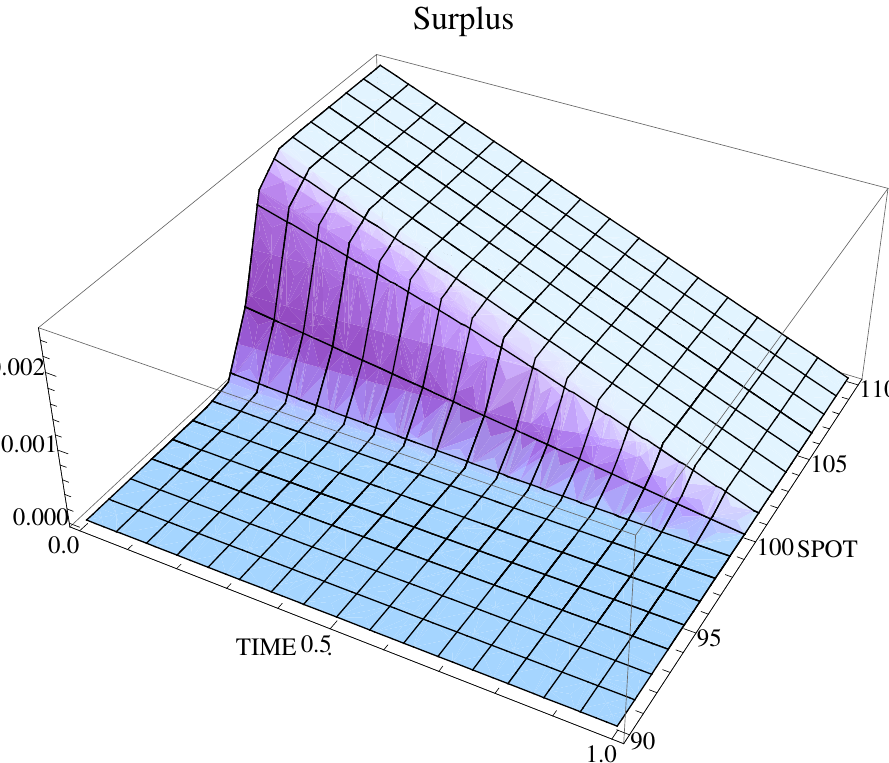} \\
        \caption[]{\small Trading speed and surplus for a risk neutral investor holding a European Call option.}
        \label{bild_1p_rn_call}
    \end{figure}

    \begin{figure}
        \includegraphics[width = 0.45 \textwidth]{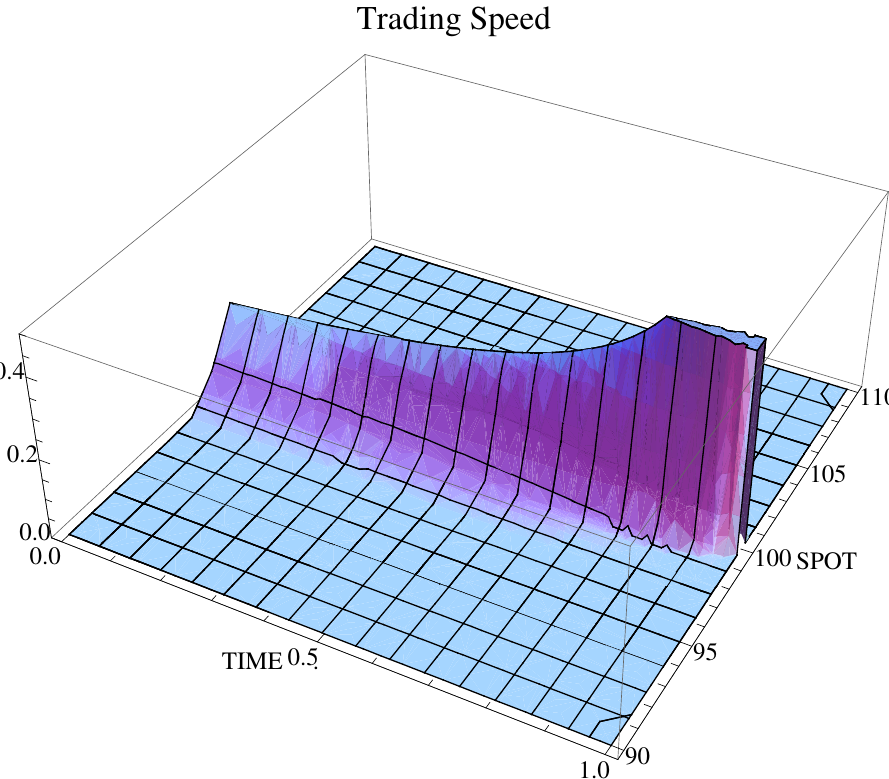} \hfill
        \includegraphics[width = 0.45 \textwidth]{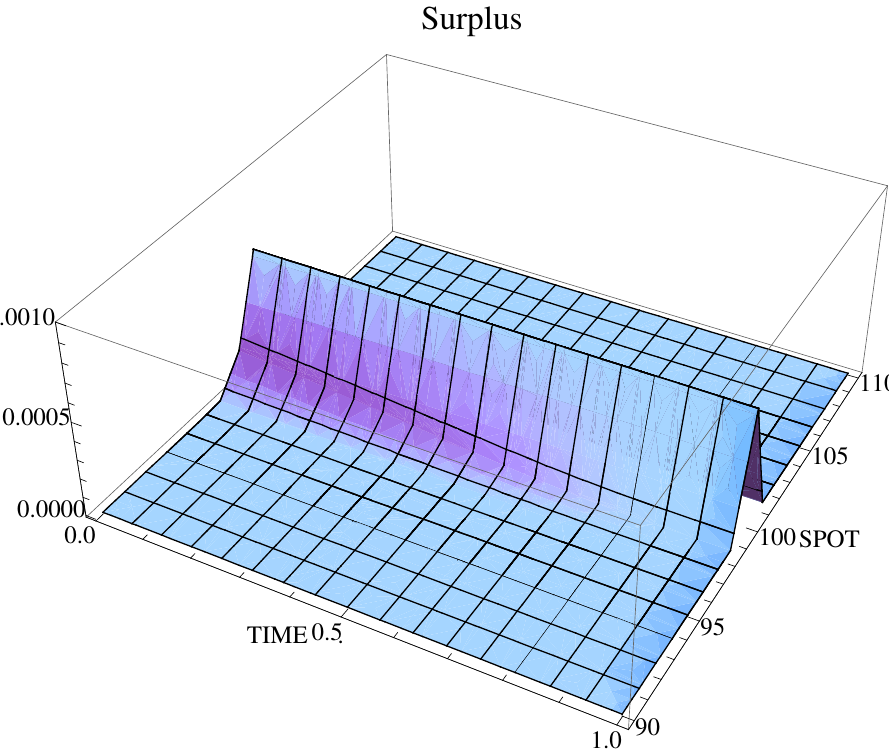} \\
        \caption[]{\small Trading speed and surplus for a risk neutral investor holding a Digital option.} 
        \label{bild_1p_rn_digital}
    \end{figure}

Figures \ref{bild_1p_rn_call_spread} and \ref{bild_1p_rn_digital_spread} illustrate that a high spread makes manipulation unattractive. It shows the optimal trading speed and the surplus at time $t=0$ for the Call and Digital option in the one player framework. We used the cost function
\[
    g(z) = \kappa z + s \cdot \sign (z) \quad \mbox{ for different spreads } s \in \{ 0, 0.001, 0.002, 0.003, 0.004 \}
\]
with the remaining parameters as above. We see that the higher the spread, the smaller the trading speed and the surplus. This is intuitive as frequent trading, in particular, when the option is at the money, incurs high spread crossing costs. The same is true for fixed transaction costs which also discourage frequent trading.

    \begin{figure}
        \includegraphics[width = 0.45 \textwidth]{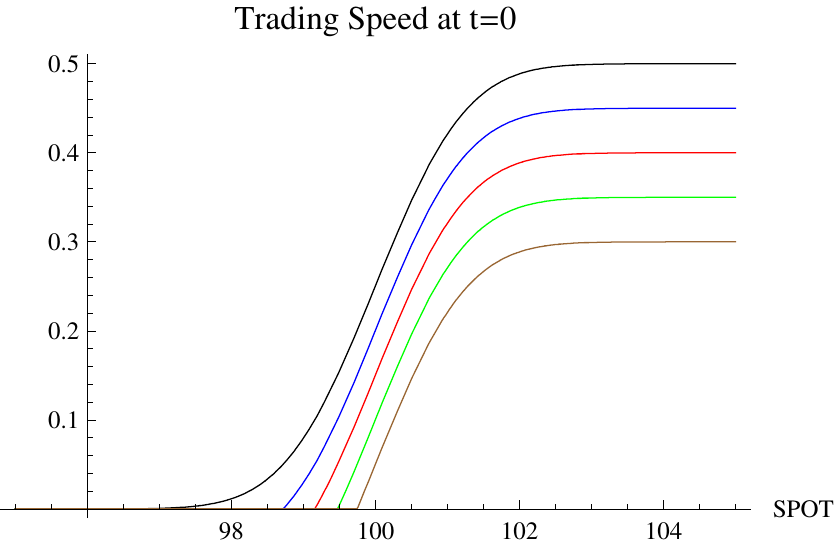} \hfill
        \includegraphics[width = 0.45 \textwidth]{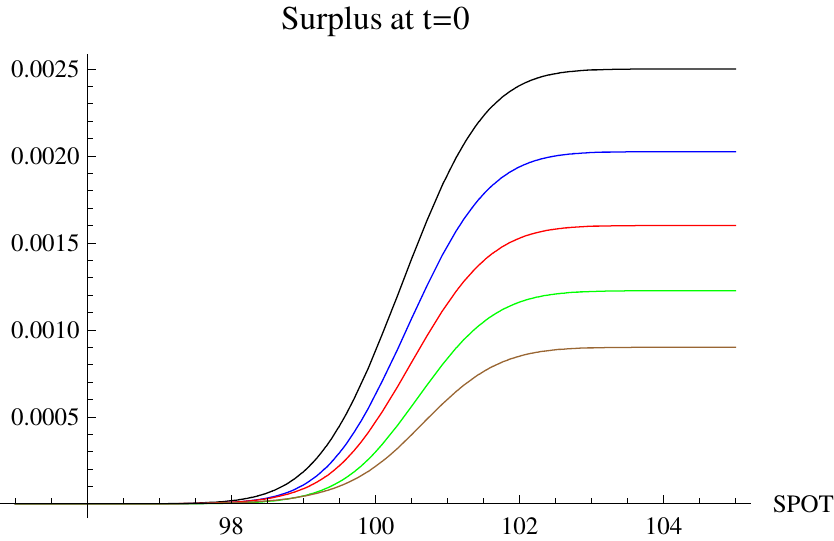} \\
        \caption[]{\small Trading speed and surplus for a risk neutral investor holding a European Call option for different spread sizes $s = $ 0 (black), 0.001 (blue), 0.002 (red), 0.003 (green), 0.004 (brown). The higher the spread, the smaller the trading speed and the surplus. }
        \label{bild_1p_rn_call_spread}
    \end{figure}
    
    \begin{figure}
        \includegraphics[width = 0.45 \textwidth]{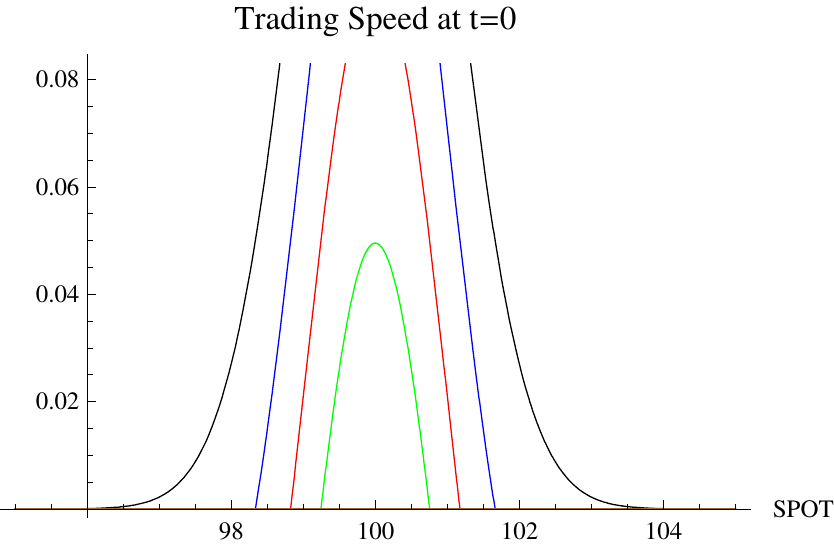} \hfill
        \includegraphics[width = 0.45 \textwidth]{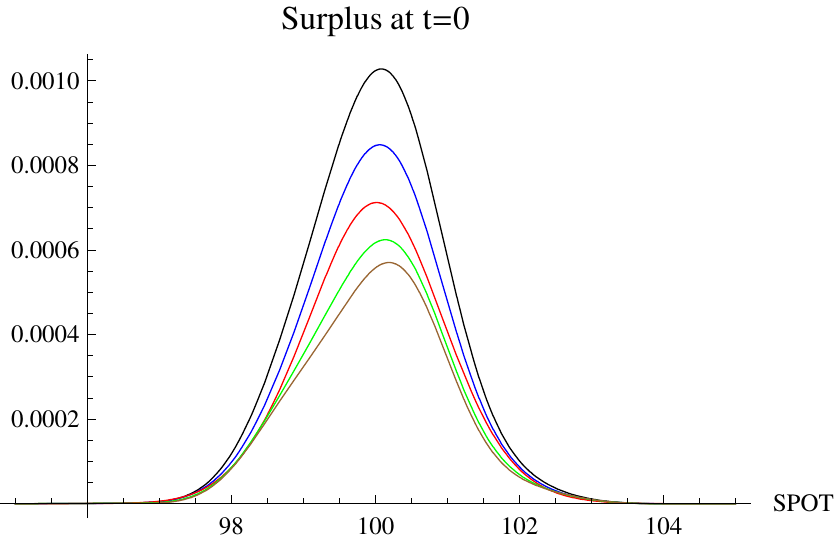} \\
        \caption[]{\small Trading speed and surplus for a risk neutral investor holding a Digital option for different spread sizes $s = $ 0 (black), 0.001 (blue), 0.002 (red), 0.003 (green), 0.004 (brown). The higher the spread, the smaller the trading speed and the surplus. }
        \label{bild_1p_rn_digital_spread}
    \end{figure}

    \subsection{Risk Averse Agents}
A second class which yields explicit results is those of exponential utility functions $ u^j(z) = - \exp \left( - \alpha^j z \right) $ for $ j = 1, ..., N$, where $\alpha^j > 0$ is the risk aversion coefficient. In this case the value functions satisfy
\[
    V^j(t, p, r) = \exp \left( \alpha^j r^j \right) \cdot V^j (t, p, 0)
\]
and thus $ V^j_{r^j} = \alpha^j V^j $. We suppress the state variable $r$ and write $V^j (t, p) \set V^j(t, p, 0)$. As above, we first show existence and uniqueness of a solution for a general cost structure. In a second step, we derive the closed form solution for the linear cost function in the single player framework.

    \subsubsection{General Cost Structure}
The HJB-equation (\ref{HJB}) turns into
    \begin{equation} \label{HJB_exp}
        0 = v^j_t + \frac{1}{2} \sigma^2 v^j_{pp} + \sup_{c^j} \left[ \lambda \left( c^j + \dot X^{-j} \right) v^j_p + c^j g \left( c^j + \dot X^{-j} \right) \alpha^j v^j \right]
    \end{equation}
with terminal condition $ v^j(T, p) = - \exp \left( - \alpha^j H^j(p) \right) $. We apply the linear transformation $ \tilde v^j \set - \frac{1}{\alpha^j} \log (- v^j)$ to turn the HJB equation into
    \begin{equation} \label{HJB_exp_transformed}
        0 = \tilde v^j_t + \frac{1}{2} \sigma^2 \tilde v^j_{pp} - \frac{1}{2} \sigma^2 \alpha^j \left( \tilde v^j_p \right)^2 + \sup_{c^j} \left[ \lambda \left( c^j + \dot X^{-j} \right) \tilde v^j_p - c^j g \left( c^j + \dot X^{-j} \right) \right]
    \end{equation}
with terminal condition $\tilde v^j (T, p) = H^j(p)$. Note that this equation equals the HJB-equation (\ref{HJB_rn}) in the risk neutral setting, up to the quadratic term $ - \frac{1}{2} \sigma^2 \alpha^j \left( \tilde v^j_p \right)^2 $. As in (\ref{strategy_j_rn}), the optimal trading speeds are
    \begin{equation} \label{strategy_j_cara}
        c^j = \dot X^j = - \frac{1}{g' \left( \dot X^* \right) } \left[ - \lambda \tilde v^j_p + g ( \dot X^* ) \right]
    \end{equation}
where the aggregate trading speed $\dot X^*$ is the unique solution to
    \begin{equation} \label{cumspeed_cara}
        0 = \lambda \sum_{i=1}^N \tilde v_p^i - N g \left( \sum_{i = 1}^N \dot X^i_t \right) - \left( \sum_{i = 1}^N \dot X^i_t \right) g' \left( \sum_{i = 1}^N \dot X^i_t \right) .
    \end{equation}
If we plug $\dot X^*$ and $\dot X^j$ back into (\ref{HJB_exp_transformed}), we get
    \begin{equation} \label{system_cara}
        0 = \tilde v^j_t + \frac{1}{2} \sigma^2 \tilde v^j_{pp} - \frac{1}{2} \sigma^2 \alpha^j \left( \tilde v^j_p \right)^2 + \lambda \dot X^* v^j_p - \dot X^j g \left( \dot X^* \right).
    \end{equation}
We can show existence and uniqueness of a solution.

    \begin{prop} \label{prop_existence_cara}
        Let $H^j \in \mathcal C^2_b$ for each $j \in J$. The Cauchy problem \ref{HJB_exp_transformed} admits a unique solution, which coincides with the vector of value functions (up to an exponential transformation).
    \end{prop}

    \begin{proof}
        See appendix A.
    \end{proof}

    \subsubsection{Linear Cost Structure, Single Player}
For the one player case with linear cost structure, we have an explicit solution:
    \begin{cor} \label{cor_cara}
        Let $N=1$ and $g(z) = \kappa z$. Then the Cauchy problem \ref{HJB_exp_transformed} admits a unique solution, which can be given in closed form.
    \end{cor}

    \begin{proof}
The maximizer in (\ref{HJB_exp_transformed}) is
    $$ c = \dot X = \frac{\lambda}{ 2 \kappa } \tilde v_p $$
and the Cauchy problem (\ref{system_cara}) turns into
    $$ 0 = \tilde v_t + \frac{1}{2} \sigma^2 \tilde v_{pp} + \left( \frac{\lambda^2 }{4 \kappa } - \frac{1}{2} \sigma^2 \alpha \right) \tilde v_p^2  $$
with terminal condition $v(T,p) = H(p) $. This is Burgers' equation. Its explicit solution is given in Lemma \ref{Burgers}.
    \end{proof}

    \subsubsection{Numerical Illustrations} \label{illustration_cara}
Let us conclude this section with numerical illustrations. We simulated the system (\ref{HJB_exp}) for two players. Figure \ref{bild_2p_cara_call} shows the aggregate optimal trading speed $\dot X(0, p) + \dot Y(0, p)$ and the surpluses $v^j(0, p) - \EE \left[ u^j \left( H( P_T ) \right) | P_0 = p \right] $ for time $t=0$ and different spot prices $p \in [95,105]$ for the European Call option $H(P_T) = (P_T - K)^+$; we assume that Player 1 (blue) is the option writer and Player 2 (red) the option issuer. We chose the strike $K = 100$, maturity $T = 1$, volatility $\sigma = 2$ and liquidity parameters $\lambda = \kappa = 0.01$ and risk aversion parameters $\alpha^1=0.01, \alpha^2 = 0.01$ (plain), respectively, $\alpha^1 = 0.001, \alpha^2 = 0.1$ (dashed). Since Player 1 has a long position in the option, she has an incentive to buy the underlying; for the same reason Players 2 has an incentive to sell it (Panel (b)). Our simulations suggest that the dependence of the equilibrium trading speed on the agents' risk aversion is weak (Panels (b) and (c)) and that overall the option issuer is slightly more active than the option writer. Furthermore, we see from Panel (d) that the issuer benefits more from reducing her loss than the writer benefits from increasing her gains. This effect is due to the concavity of the utility function and increases with the risk aversion.

        \begin{figure}
                \subfloat[]{ \includegraphics[width = 0.45 \textwidth]{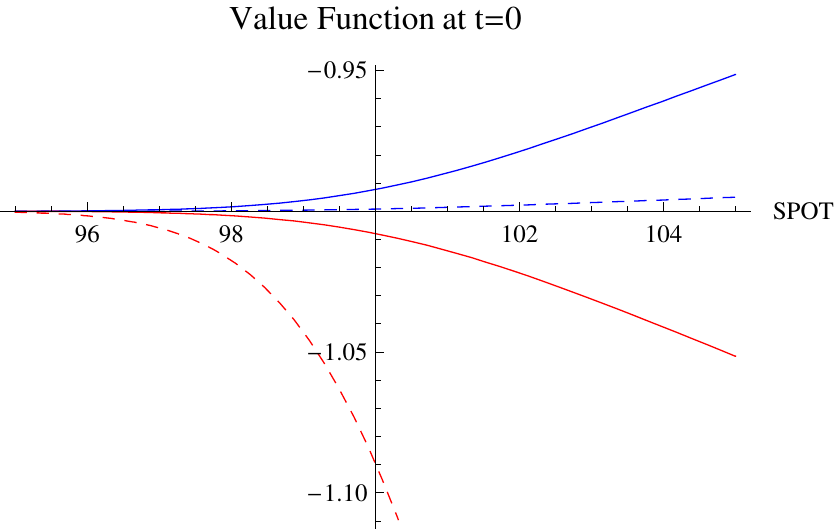} } \hfill
                \subfloat[]{ \includegraphics[width = 0.45 \textwidth]{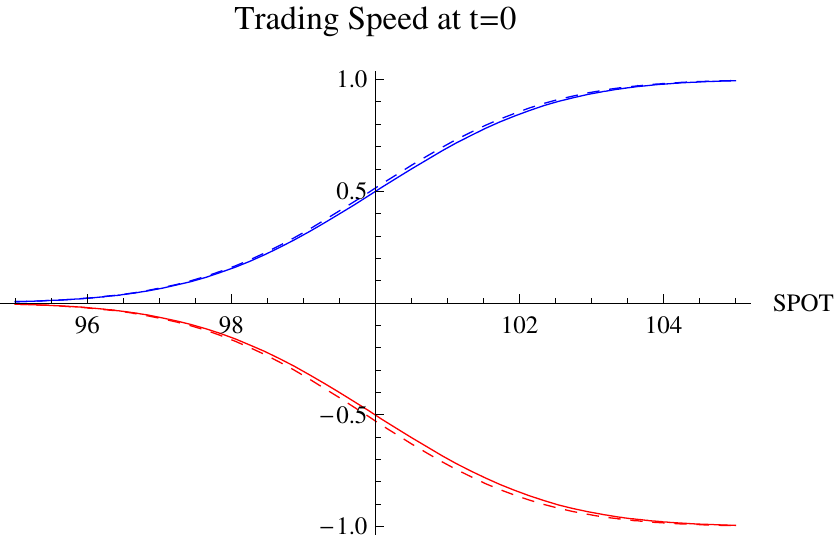} } \\
                \subfloat[]{ \includegraphics[width = 0.45 \textwidth]{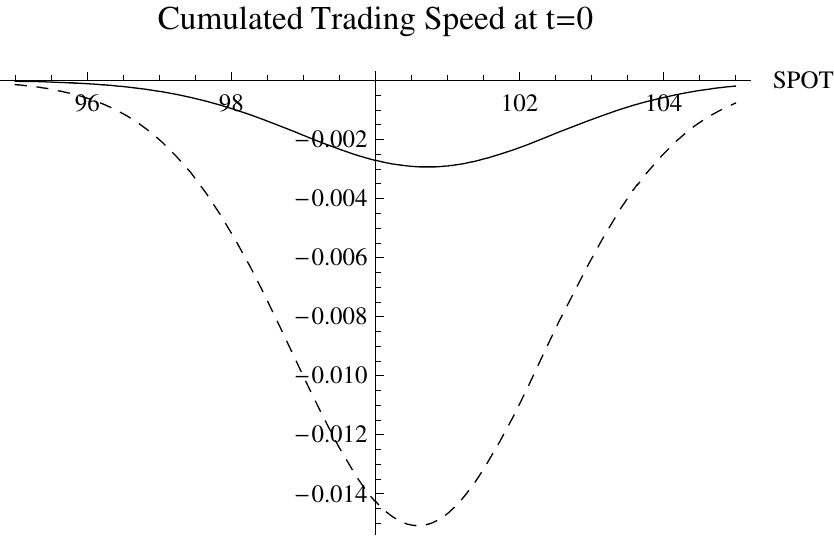} } \hfill
                \subfloat[]{ \includegraphics[width = 0.45 \textwidth]{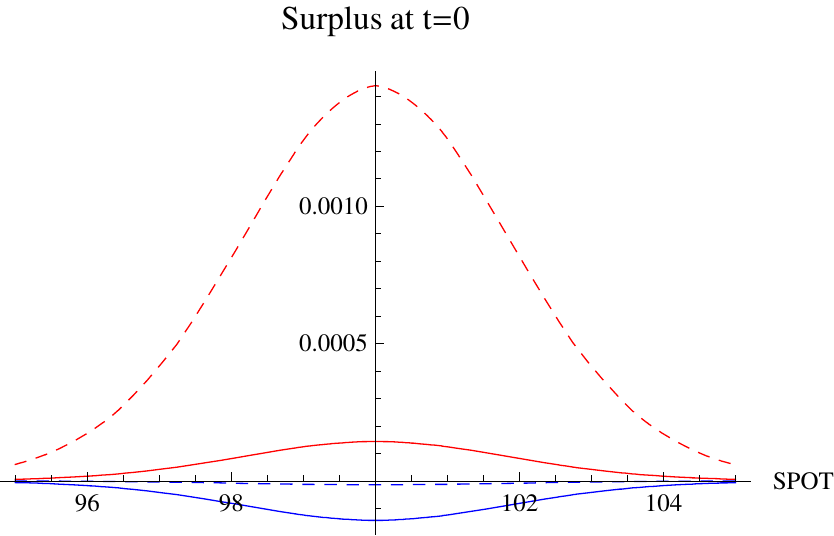} } \\
                \caption[]{\small Value function, trading speed, aggregate trading speed and surplus for the writer (blue) and issuer (red) of a European Call option when both agents are risk averse. The plain (dashed) curves display the case where issuer is about as (more) risk averse than the option writer.}
               \label{bild_2p_cara_call}
        \end{figure}

    \section{How to Avoid Manipulation}
In this section, we use the closed form solutions for risk neutral agents derived in subsection \ref{riskneutral} to illustrate how an option issuer may prevent other market participants from trading against her by using their impact on the dynamics of the underlying. Some of our observations were already made in Gallmeyer \& Seppi (2000) in a three-period binomial model. We start with the simplest case of a zero-sum game.

\begin{cor}
    Let all players be risk neutral with offsetting payoffs $ \sum_{i=1}^N H^i = 0 $. Then the aggregate trading speed is $ \sum_{i=1}^N \dot X^i \equiv 0 $.
\end{cor}
\begin{proof}
    If $ \sum_{i=1}^N H^i = 0 $ then (\ref{speed}) and (\ref{cum_solution_rn}) imply that the aggregate value function $v = \sum_{i=1}^N v^i$ equals zero. It follows from (\ref{speed}) that the aggregate trading speed vanishes.
\end{proof}

Loosely speaking in a zero-sum game, if all option traders are risk neutral and willing to move the market in their favor, their combined effect cancels. We note that this is no longer true for general utility functions, as illustrated in figure \ref{bild_2p_cara_call}. Of course, in reality market manipulation is illegal and  many investors are unable to manipulate the underlying in the first place. This is why we now look at the following asymmetric situation: The option issuer, Player 0, does not trade the underlying; her competitor, Player 1, owns the derivative $H^1 \neq 0$ and intends to move the stock price to her favor. In addition, there are $N-1$ informed investors without option endowment in the market. They are ''predators'' that may supply liquidity and thus reduce the first player's market impact, cf. Carlin, Lobo \& Viswanathan (2007) and Schoeneborn \& Schied (2007). We find that the more informed competitors are active, the less aggregate manipulation will occur.

\begin{cor}
    Let $H^1 \neq 0$ and $ H^i = 0 $ for $ i = 2, ..., N$. Then $ \lim_{ N \to \infty} \sum_{i=1}^N \dot X^i_t = 0 . $
\end{cor}
\begin{proof}
    The solution $v$ to the Cauchy problem (\ref{cum_solution_rn}) with terminal condition
        $$ v(T,p) = \sum_{i=1}^N H^i(p) = H^1(p) $$
    is bounded uniformly in $N$. Equation (\ref{speed}) yields
        $$ \sum_{i=1}^N \dot X^i = \frac{\lambda}{\kappa} \frac{1}{N+1} v_p \stackrel{N \to \infty}{\rightarrow} 0 . $$
\end{proof}

Let us modify the preceding setting a little. Again, Player 0 issues a product $H$ and does not intend to manipulate the underlying, while her competitors do. More precisely, assume that player $0$ splits the product $H$ into pieces and sells them to $N$ risk neutral competitors, such that each of them gets $\frac{1}{N} H$. We find that their aggregate trading speed $\sum_{i=1}^N \dot X^i$ is decreasing in the number of competitors $N$. Consequently, the option issuer should sell her product to as many investors as possible in order to avoid being outsmarted. We illustrate this result in figure \ref{Np_rn_call_cumspeed}, which shows the aggregate trading speed at time $t=0$ of $N$ players each holding $1/N$ option shares.

\begin{cor}
    Let $H^i = \frac{1}{N} H$ for $i = 1,..., N$. Then for each $t \in (0,T]$ the aggregate trading speed $ \sum_{i=1}^N \dot X^i_t $ is decreasing in $N$ and $ \lim_{ N \to \infty } \sum_{i=1}^N \dot X^i_t = 0 . $
\end{cor}

\begin{proof}
    As in the preceding corollary, the assertion follows from equations (\ref{cum_solution_rn}) and (\ref{speed}).
\end{proof}

    \begin{figure}
        \subfloat[Call]{ \includegraphics[width = 0.5 \textwidth]{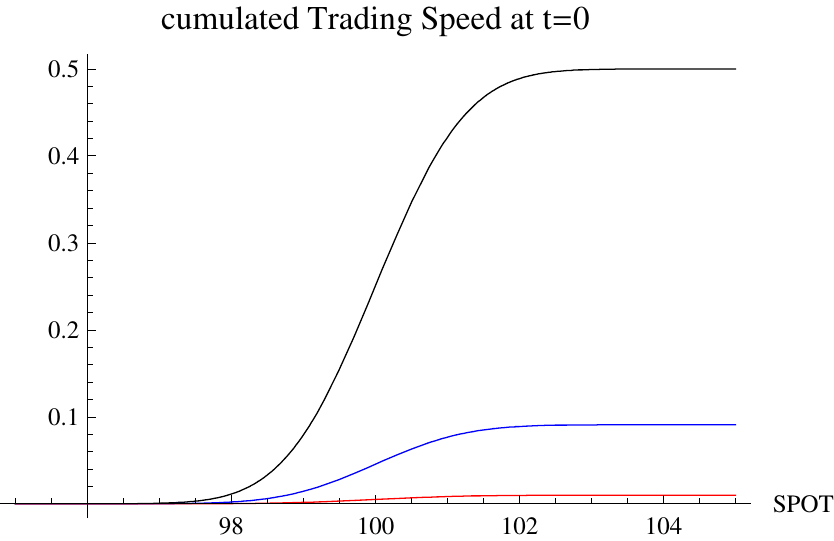} } \hfill
        \subfloat[Digital]{ \includegraphics[width = 0.5 \textwidth]{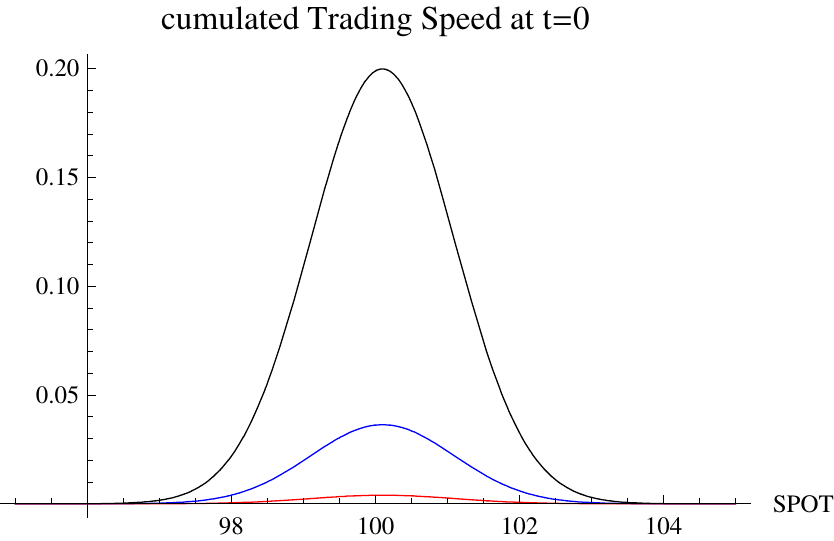} } \\
        \caption[]{\small Aggregate trading speed $\dot X^*$ at time $t=0$ for $N = $ 1 (black), 10 (blue), 100 (red) players each holding $1/N$ shares of a Call (left) and Digital (right) option with strike $K=100$. The more agents, the less aggregate manipulation.}
        \label{Np_rn_call_cumspeed}
    \end{figure}

The preceding results indicate how an option issuer can prevent her competitors from manipulation. One strategy is public announcement of the transaction: the more informed liquidity suppliers on the market, the smaller the impact on the underlying. A second strategy is splitting the product into pieces; the more option writers, the less manipulation. Let us conclude this section with a surprisingly simple way to avoid manipulation: using options with physical delivery.

\begin{rem}
    Calls, Puts and Forwards with physical delivery do not induce stock price manipulation.
\end{rem}

\begin{proof}
    Consider a risk neutral agent who owns $\Theta > 0$ Call options with physical delivery and strike $K$. As above, we denote by $X$ her strategy in the underlying. At maturity, she exercises $ 0 \leq \theta \leq \Theta $ of her Call options. Problem (\ref{problem}) turns into:
    \begin{eqnarray*}
        & & \sup_{ X \in \fX, \theta \leq \Theta } \EE \left[ \int_0^T - \dot X_t ( P_t + \kappa \dot X_t) dt + (X_T + \theta) \left( P_T - \frac{1}{2} \lambda (X_T + \theta) \right) - \theta K \right] \\
        &=& \sup_{ X \in \fX, \theta \leq \Theta } \EE \left[ \int_0^T - \kappa \dot X_t^2 dt + \theta \left( P_T - \frac{1}{2} \lambda \theta \right) - \theta K  \right] \\
        &=& \sup_{ X \in \fX } \EE \left[ \int_0^T - \kappa \dot X_t^2 dt \right] + \sup_{ \theta \leq \Theta} \left[ \theta \left( P_T - \frac{1}{2} \lambda \theta \right) - \theta K \right] \\
        &=& \sup_{ \theta \leq \Theta} \left[ \theta \left( P_T - \frac{1}{2} \lambda \theta \right) - \theta K \right]
    \end{eqnarray*}
    where the first term in the first line describes the expected trading costs in $[0,T]$ and the liquidation value of $\theta + X_T $ stock shares at maturity. The optimal trading strategy is $X \equiv 0$. The same holds true for Put options and Forward options with physical delivery.
\end{proof}

    \section{Conclusion}
We investigated the strategic behavior of option holders in illiquid markets. If trading the underlying has a permanent impact on the stock price, the possession of derivatives with cash delivery may induce market manipulation. We showed the existence and uniqueness of optimal trading strategies in continuous time and for a general cost function; in the one player framework for general utility functions, and in the multi-player case for risk neutral as well as CARA investors. Moreover, we showed how market manipulation can be reduced.

Our work may be extended in several directions. Foremost, we derived our results under Assumption \ref{assumption}. This assumption is only satisfied in the single-player risk-neutral case where the expected costs of buying a portfolio over a finite time interval under market impact equals its expected liquidation value under infinitely slow liquidation and does not hold in general. The problem of defining a proper notion of liquidation value under strategic interaction is important but was not our focus and is left for future research. Furthermore, it would be interesting to consider American or path-dependent options as well as more sophisticated market impact models such as Obizhaeva \& Wang (2006) that account for resilience effects and temporary price impacts.

\begin{appendix}

    \section{An Existence Result}
In this section, we prove Propositions \ref{prop_existence_rn} and \ref{prop_existence_cara} where the PDE (\ref{system_rn}) in the risk neutral setting is a special case of the system (\ref{HJB_exp_transformed}) for risk averse agents, with $\alpha^j = 0$ for each $j$. In order to establish our existence and uniqueness of equilibrium result, we adopt the proof of Proposition 15.1.1 in Taylor (1997) to our framework. After time inversion from $t$ to $T - t$ both systems of PDEs are of the form
    \begin{equation} \label{PDE}
        v_t = Lv + F(v_p)
    \end{equation}
for $v \set \left( v^1, ..., v^N \right) $, where $L$ is the Laplace-operator
    $$ L = \frac{1}{2} \sigma^2 \frac{\partial^2}{\partial p^2} $$
and $F = \left( F^1, ..., F^N \right) $ is of the form
    $$ F^j(v_p) = - \frac{1}{2} \sigma^2 \alpha^j \left(  v^j_p \right)^2 + \lambda \dot X^* v^j_p - \dot X^j g \left( \dot X^* \right). $$
Here $\dot X^*$ and $\dot X^j$ are given implicitly by (\ref{strategy_j_rn}) and (\ref{cumspeed_rn}). The initial condition is
    \begin{equation} \label{initial}
        v(0, p) = H (p) = \left( H^1, ..., H^N \right) .
    \end{equation}
We rewrite (\ref{PDE}) in terms of an integral equation as
    \begin{equation} \label{PDE_transf}
        v(t) = e^{t L} + \int_0^t e^{(t-s)L} F(v_p(s)) ds \set \Psi v(t) .
    \end{equation}
and seek a fixed point of the operator $\Psi$ on the following set of functions:
    $$ \XX = \mathcal C^1_b (\RR, \RR^N) \set \left\{ v \in \mathcal C^1 (\RR, \RR^N) \mid v, v_p \text{ bounded} \right\} $$
equipped with the norm
    $$ \norm{v}_{\XX} \set \norm{v}_{\infty} + \norm{v_p}_{\infty}  . $$
We set $\YY \set \mathcal C_b$. Note that $\XX$ and $\YY$ are Banach spaces and the semi-group $e^{tL}$ associated with the Laplace operator is strongly continuous on $\XX$, sends $\YY$ on $\XX$ and satisfies
\[
    \norm{e^{t L}}_{\mathcal L ( \YY, \XX ) } \leq C t^{-\gamma}
\]
for some $C > 0$, $\gamma < 1$ and $t \leq 1$. Furthermore, the nonlinearity $F$ is locally Lipschitz and belongs to $\mathcal C^{\infty}$. Indeed, the map $ a \mapsto \dot X^*(a) $ is $\mathcal C^{\infty}$, due to the implicit function theorem with first derivative
        $$ \frac{\partial}{\partial v_p} \dot X^* (v_p) = \frac{ \lambda }{ (N+1) g'( \dot X^* (v_p) - \dot X^* (v_p) g''(\dot X^* (v_p)) } $$
    where the denominator is positive due to Assumption \ref{assumption_g}. The cost function $g$ is $\mathcal C^{\infty}$ by assumption. In particular, the assumptions of Proposition 15.1.1 in Taylor (1997) are satisfied.

The a-priori estimates of Proposition \ref{boundedstrategy} yield that, if a solution $v$ to (\ref{PDE}) exists, it is bounded in the sense $\norm{v}_{\XX} \leq K$. Therefore, we define
\[
    \XX_K \set \left\{ v  \in \XX \mid \norm{ v }_{\XX} \leq K \right\}
\]
and choosing $K$ large enough we may assume that the initial condition satisfies $H \in \XX_K$.

We are now ready to prove existence and uniqueness of a solution to (\ref{PDE_transf}). Proposition 15.1.1 in Taylor (1997) gives a solution for a small time horizon $[0, \tau]$, with $\tau>0$ specified below. We apply his argument recursively to extend the solution to $[0, T]$.

    \begin{prop}
        There is $\tau > 0$ such that for each $n \in \NN_0$, the PDE (\ref{PDE_transf}) with initial condition (\ref{initial}) admits a unique classical, bounded solution in $\XX_K$ on the time horizon $[0, n \tau \wedge T]$. This solution coincides with the value function.
    \end{prop}

\begin{proof}
    \begin{enumerate}
        \item
            For $n = 0$, there is nothing to prove. Pick $n \in \NN$ such that $n \tau < T$. By induction, we can assume that there is a solution $v^{(n)} \in \XX_K$ on the time horizon $[0, n \tau]$. In particular, the initial condition for the next recursion step $h^{(n)} \set v^{(n)} (n \tau)$ is in $\XX_K $.
        \item
        Fix $\delta > 0$. We construct a short time solution on the following set of functions:
            $$ Z^{(n+1)} \set \left\{ v \in \mathcal C \left( [ n \tau , (n+1) \tau ], \XX \right) \mid v( n \tau ) = h^{(n)}, \norm{u(t) - h^{(n)}}_{\XX} \leq \delta \; \forall t \right\} . $$
        To this end, we first show that $ \Psi: Z ^{(n+1)} \to Z ^{(n+1)}$ is a contraction, if $\tau > 0$ is chosen small enough.

                        For this, let $\tau_1$ be small enough such that for $ t \leq \tau_1$ and any $v \in \XX_K$ we have
                            $$ \norm{e^{t L} v - v }_{\XX} \leq \frac{1}{2} \delta . $$
                        Here we used that $e^{t L}$ is a continuous semigroup and $\norm{ v }_{\XX} \leq K $. In particular, for $v = h^{(n)}$:
                            $$ \norm{e^{t L} h^{(n)} - h^{(n)} }_{\XX} \leq \frac{1}{2} \delta . $$

                        For $v \in Z ^{(n+1)}$, the derivative $v_p$ is uniformly bounded in the sense $ \norm{v_p}_{\infty} \leq \norm{ h^{(n)} }_{\XX} + \delta \leq K + \delta$. Hence, we only evaluate $F$ on compact sets. By assumption, $F$ is locally Lipschitz. In particular, $F$ is Lipschitz on compact sets. In other words, there is a constant $K_1$ such that for any $v, w \in Z^{(n+1)}$ we have
                                $$ \norm{ F(v_p) - F(w_p) }_{\YY} \leq K_1 \norm{ v - w }_{\XX} $$
                        This implies, for $ w = h^{(n)} $
                            \begin{eqnarray*}
                                \norm{F(v_p)}_{\YY} &\leq & \norm{F( h^{(n)}_p )}_{\YY} + K_1 \norm{ v - h^{(n)} }_{\XX} \\
                                            & \leq & K + K_1 \delta \\
                                            & \set & K_2 .
                            \end{eqnarray*}

                        This, together with the boundedness assumption on $e^{t L}$, yields
                            \begin{eqnarray*}
                                \norm{ \int_{n \tau}^t e^{(t-y)L}  F( v_p (y) ) dy }_{\XX}
                                    & \leq &    t \norm{e^{t L}} \sup_{n \tau \leq y \leq t} \norm{ F(v_p(y)) }_{\YY} \\
                                    & \leq & t^{1-\gamma} C K_2 .
                            \end{eqnarray*}
                        This quantity is $\leq \frac{1}{2} \delta$ if $t \leq \tau_2 \set \left( \frac{\delta}{2 C K_2} \right)^{\frac{1}{1-\gamma}} $.

                        Finally, it follows that for $v \in Z ^{(n+1)}$ we have
                            \begin{eqnarray*}
                                \norm{\Psi v - h^{(n)} }_{\XX}
                                    & \leq & \norm{ e^{t L} h^{(n)} - h^{(n)} }_{\XX} + \norm{ \int_{n \tau}^t e^{ (t-y)L}  F(v_p(y) ) dy }_{\XX} \\
                                    & \leq & \frac{1}{2} \delta + \frac{1}{2} \delta = \delta .
                            \end{eqnarray*}
                        This shows that $\Psi$ maps $Z ^{(n+1)} $ into itself.

                        It remains to show that $\Psi$ is a contraction. Let $v, w \in Z ^{(n+1)}$. Then
                            \begin{eqnarray*}
                                \norm{ \Psi v(t) - \Psi w(t)}_{\XX} &=& \norm{ \int_{n \tau}^t e^{(t-y)L} \left[ F (v_p(y)) - F (w_p(y)) \right] dy }_{\XX} \\
                                    & \leq & t \norm{ e^{t L} } \sup_{n \tau \leq y \leq t} \norm{ F (v_p(y)) - F (w_p(y)) }_{\YY} \\
                                    & \leq & t^{1-\gamma} C K_2 \sup_{n \tau \leq y \leq t} \norm{v(y) - w(y)}_{\XX}
                            \end{eqnarray*}
                            The quantity $ t^{1-\gamma} C K_2$ is $\leq \frac{1}{2}$ if $ t \leq \tau_3 \set \left( \frac{1}{ 2 C K_2} \right)^{\frac{1}{1-\gamma}}$.
                        This proofs that $\Psi$ is a contraction in $Z ^{(n+1)} $, if $\tau$ is small in the sense
                            $$ 0 < \tau \set \min \{ \tau_1, \tau_2, \tau_3 \} . $$
                        Note that the time step $\tau$ does not depend on $n$. It is the same in every recursion step.

        \item It follows that $\Psi$ has a unique fix point $v $ in $Z ^{(n+1)} $. In other words, we constructed a function $v \in \mathcal C ( [n \tau, (n+1) \tau], \XX ) = \mathcal C^{0,1} [n \tau , (n+1) \tau] $ which solves the PDE (\ref{PDE_transf}) with initial condition $v(s) = h^{(n)} = v^{(n)} (n \tau)$ on the time interval $ [n \tau, (n+1) \tau] $.

        This solution is actually in $\mathcal C^{1,2} \left( ( n \tau, (n+1) \tau ] \times \RR, \RR^N \right) $, due to Proposition 15.1.2 in Taylor (1997). Furthermore, $v$ is bounded by construction. Indeed, $\norm{v}_{\infty} \leq \norm{ h^{(n)} }_{\XX} + \delta \leq K + \delta $. We define the new solution as
            $$ v^{(n+1)} \set v^{(n)} \mathds 1_{ \{ 0 \leq t \leq n \tau \} } + v \mathds 1_{ \{ n \tau < t \leq (n+1) \tau \} } . $$
        By construction, $v^{(n+1)}$ solves (\ref{PDE_transf}) on the time horizon $[0, (n+1) \tau]$ and is bounded and in $\mathcal C^{1,2}$. Hence, we can apply the Verification Theorem IV.3.1 from Fleming \& Soner (1993), which yields that $ v^{(n+1)} $ coincides with the vector of value functions (up to time reversal and an exponential transformation, if $\alpha^j > 0$). The a priori estimate in Proposition \ref{boundedstrategy} yields that $v^{(n+1)} \in \XX_K$. In particular, $ \norm{ v^{(n+1)} ( (n+1) \tau) }_{\XX} \leq K$, which is necessary for the next recursion step.
    \end{enumerate}
    This completes the proof.
\end{proof}

\end{appendix}


\end{document}